\documentclass{cccg20}
\usepackage{graphicx,amssymb,amsmath}

\usepackage{enumitem}
\usepackage{comment}
\newtheorem{subproblem}{SubProblem}

\usepackage{amsmath}
\usepackage[utf8]{inputenc}
\usepackage{algorithm}
\usepackage[noend]{algpseudocode}
\usepackage{subcaption}
\usepackage[colorinlistoftodos]{todonotes}
\usepackage{booktabs}
\usepackage{xcolor}
\definecolor{armygreen}{rgb}{0.29, 0.33, 0.13}

\usepackage{hyperref}

\graphicspath{{./figs/}}

\title{City Guarding with Limited Field of View}

\author{Ovidiu Daescu \thanks{The University of Texas at Dallas, {\tt daescu@utdallas.edu}}
        \and
        Hemant Malik\thanks{The University of Texas at Dallas, {\tt  malik@utdallas.edu}}}

\index{Daescu, Ovidiu}
\index{Malik, Hemant}

\begin{document}
\thispagestyle{empty}
\maketitle

\begin{abstract}
Drones and other small unmanned aerial vehicles are starting to get permission to fly within city limits. 
While video cameras are easily available in most cities, their purpose is to guard the streets at ground level.
Guarding the aerial space of a city with video cameras is a problem that so far has been largely ignored. 

In this paper, we present bounds on the number of cameras needed to guard a city’s aerial space (roofs, walls, and ground) using cameras with 180$^{\circ}$ range of vision (the region in front of the guard), which is common for most commercial cameras.
We assume all buildings are vertical and have a rectangular base. Each camera is placed at a top corner of a building.

We considered the following two versions: (i) buildings have an axis-aligned ground base and, (ii) buildings have an arbitrary orientation. 
We give necessary and sufficient results for (i), necessary results for (ii), and conjecture sufficiency results for (ii).  
Specifically, for (i) we prove a sufficiency bound of $2k + \lfloor \frac{k}{4} \rfloor + 4$ on the number of vertex guards, while for (ii) we show that $3k+1$ vertex guards are sometimes necessary, where $k$ is total number of buildings in the city. 


\end{abstract}

\section{Introduction}

Drones and other small unmanned aerial vehicles (UAVs) are already allowed to experimentally fly within 
city limits. For example, in August 2019, Uber announced it has selected the city of Dallas to experiment with flying drones and small UAVs, within the city. Monitoring the aerial space of big cities is thus becoming a critical problem that yet has to be addressed. Video cameras are easily available in most cities, but their purpose is to guard the streets at ground level. Guarding the aerial space of a city with cameras is a problem that has been largely ignored.    

City guarding is related to the famous \textit{art gallery problem}~\cite{klee1969every} and its many variations~\cite{variation} studied in the past few decades. In almost all these studies, the art gallery lies in the plane (2D), assuming a polygonal shape with or without holes. In the art gallery problem, the goal is to determine the minimum number of point guards sufficient to see every point of the interior of a simple polygon. A point \textit{q} is visible to guard \textit{g} if the line segment joining \textit{q} and \textit{g} lies completely within the polygon.
When the guards are restricted to vertices of the polygon only, they are referred to as vertex guards. 

In the orthogonal art gallery problem, all edges of the polygon are either horizontal or vertical. In some versions of the art gallery problem, the polygon is allowed to have $h$ holes. When guarding such polygons, it is allowed to place the guards at the vertices of the enclosing polygon and the vertices of the holes.  

For guarding an \textit{orthogonal polyhedron} point guards are less effective. There exist examples of polyhedra with \textit{n} vertices where guards placed at every vertex do not cover the whole interior of the polyhedra; instead $O(n^{3/2})$ non-vertex guards are required~\cite{artgallery}. 

The problem is also related to the following problem~\cite{blanco1994illuminating}: Given $k$ pairwise disjoint isothetic rectangles
in a plane, place vertex guards on rectangles such that every point in free space (plane area excluding the interior of the quadrilaterals) is visible to at least one guard. 

The city guarding problem was introduced in~\cite{cityguarding}, and is a 2.5D variant of the 2D orthogonal art gallery with holes. The input consists of $k$ buildings, within an area bounded by an axis-parallel rectangle (this can be relaxed to the whole plane, assuming cameras have unlimited distance visibility), with each building being vertical and having an axis parallel rectangular base (a vertical rectangular prism), and the goal is to place the minimum number of guards that can see in any direction (referred as 360$^{\circ}$ field of vision), at the top corners (vertices) of some buildings, to guard the aerial space of the city. The height of a building is a strictly positive real number. In~\cite{cityguarding}, they consider three variations of city guarding: (i) \textit{Roof Guarding:} determine the minimum number of vertex guards required to guard the roofs (ii) \textit{Ground and Wall Guarding:} determine the minimum number of vertex guards necessary to guard the ground and the walls, and (iii) \textit{City Guarding:} determine the minimum number of vertex guards required to guard the (aerial space of the) city, which means the roofs, walls, and the ground. As with the 2D art gallery problem, the 2.5D city guarding problems are NP-hard and, by a simple reduction, so are the corresponding versions studied in this paper. 

We consider the three variations of the city guarding problem with a restriction on the visibility range of the guards. Specifically, a guard is only able to see the region in front of it, i.e., the range of vision of a guard is bounded by 180$^{\circ}$, instead of the 360$^{\circ}$ in~\cite{cityguarding}. This corresponds to the capabilities of most commercial cameras.

In all our proofs each guard is placed at the top corner of a building and is oriented such that the seen and unseen regions of the guard are separated
by a vertical plane parallel with one of the sides of the building where the camera is placed. Thus, when building bases are isothetic (axis-aligned) rectangles, a camera will face in one of four directions: East, West, North, or South (E, W, N, S). From now on, we assume cameras are placed as stated here, unless otherwise specified, and may omit mentioning camera orientation 
throughout the paper. 

The two versions we consider are: (A) Buildings have an axis-aligned rectangular base (isothetic rectangles), (B) Buildings have a (arbitrary oriented) rectangular base. To solve the two versions, we address the following variations of the art-gallery problem:

\noindent (V1) Given an axis-aligned rectangle $P$ with $k$ disjoint axis-aligned rectangular holes, place vertex guards on hole boundaries such that every point inside $P$ is visible to at least one guard, where the range of vision of a guard is $180^{\circ}$.

\noindent (V2) Given an axis-aligned rectangle $P$ with $k$ disjoint (arbitrary oriented) rectangular holes, place vertex guards on hole boundaries such that every point inside $P$ is visible to at least one guard, where the range of vision of a guard is $180^{\circ}$.


For the first problem (V1), we prove a sufficiency bound of $2k + \lfloor \frac{k}{4} \rfloor + 4$ on the number of vertex guards. To obtain this bound, we provide a novel, divide and conquer algorithm. 
For the second problem (V2), we show that $3k+1$ vertex guards are sometimes necessary and conjecture that the bound is tight. 

A comparison of our sufficiency and necessity results with those in~\cite{cityguarding} is shown in Table~\ref{t:results}. Our solutions 
set an essential foundation for monitoring drones flying within city limits, using video cameras.


\begin{small}
\begin{table}[ht]
 
\begin{center}
\scalebox{0.62}{%
\begin{tabular}{|p{3cm}|p{3cm}|p{3cm}|p{3cm}|} \hline
  & ~\cite{cityguarding} Guard vision range: 360$^{\circ}$ (axis-aligned rectangle buildings)& Guard vision range: 180$^{\circ}$ (axis-aligned rectangle buildings) &Guard vision range: 180$^{\circ}$ (non-axis-aligned rectangle buildings) \\ \hline
  Roof Guarding & {\color{blue}$\lfloor \frac{2(k-1)}{3} \rfloor +1$} & {\color{blue}$k$} &
 {\color{blue}$k$}\\ \hline
    Ground and Wall Guarding & {\color{red}$k + \lfloor \frac{k}{4} \rfloor +1$} & 
{\color{red}$2k + \lfloor \frac{k}{4} \rfloor + 4$} & {\color{green}$3k+1$}\\ \hline
    City Guarding & {\color{red}$k + \lfloor \frac{k}{2} \rfloor +1$} & {\color{red}$2k + \lfloor \frac{k}{4} \rfloor + 4$} & {\color{green}$3k+1$}\\ \hline
  \end{tabular}}
\end{center}

   \caption{Sufficient and necessary results comparisons. A tight bound is shown in blue color, a sufficiency bound in red color and a necessary bound in green color.}
 \label{t:results}

\end{table}
\end{small}

\section{Related Work}
\label{RW}

\subsection{Simple Polygons Results}

Given a simple polygon $P$ in the plane, with $n$ vertices, Chvatal~\cite{chvatal} proved that $\lfloor n/3 \rfloor$ vertex guards are always sufficient and sometimes necessary to guard $P$. Chvatal’s proof was later simplified by Fisk~\cite{fisk1978short} using the existence of a three-coloring of a triangulated polygon. 

When the view of the guard is limited to 180$^{\circ}$, Toth~\cite{toth2000art} showed that $\lfloor \frac{n}{3} \rfloor$ point guards are always sufficient to cover the interior of $P$ (thus, moving from 360 to 180 range of vision keeps the same sufficiency number). 
F. Santos conjecture that $\lfloor \frac{3n-3}{5} \rfloor \pi$ vertex guards are always sufficient and occasionally necessary to cover any polygon with n vertices. Later in 2002, Toth~\cite{toth2002art} provided a lower bound on the number of point guards when the range of vision $\alpha$ is less than 180$^{\circ}$. When $\alpha < 180^{\circ}$, there exist a polygon $P$ that cannot be guarded by $\frac{2n}{3} - 3$  guards. For $\alpha < 90^{\circ}$ there exist $P$ that cannot be guarded by $\frac{3n}{4} - 1$ guards, and for $\alpha < 60^{\circ}$ there exist $P$ where the number of guards needed to cover $P$ is at least $\lfloor \frac{60}{\alpha} \rfloor \frac{(n-1)}{2}$. 

\subsection{Orthogonal Polygons Results}

In 1983, Kahn et al.~\cite{kahn1983traditional} showed that if every pair of adjacent sides of the polygon form a right angle, then $\lfloor \frac{n}{4} \rfloor$ vertex guards are occasionally necessary and always sufficient to guard a polygon with $n$ vertices.

In 1983, O’Rourke~\cite{Lshaped} showed that $1 + \lfloor \frac{r}{2} \rfloor$ vertex guards are necessary and sufficient to cover the interior of an orthogonal polygon with $r$ reflex vertices. Castro and Urrutia~\cite{estivill1994optimal} provided a tight bound of $\lfloor \frac{3(n-1)}{8}  \rfloor$ on the number of orthogonal guards placed on the vertices, sufficient to cover an orthogonal polygon with $n$ vertices. 

\subsection{Polygon with Holes Results}

For a polygon $P$ with $n$ vertices and $h$ holes, the value $n$ is the sum of the number of vertices of $P$ and the number of vertices of the holes. Let $g(n, h)$ be the minimum number of point guards and $g^v(n, h)$ be the minimum number of vertex guards necessary to cover any polygon with $n$ vertices and $h$ holes. 

O'Rourke~\cite{o1983galleries} gave a first proof on guarding polygons with holes and showed that $g^v(n, h) \leq \lfloor \frac{n + 2h}{3} \rfloor$.  Shermer conjectured that $g^v(n, h) \leq \lfloor \frac{n+h}{3} \rfloor$ and this is a tight bound. He was able to prove that, for $h = 1$, $g^v(n, 1) = \lfloor \frac{n + 1}{3}  \rfloor$. However, for $h>1$ the conjecture remains open. Shermer's result can be found in~\cite{artgallery,shermer}.

Sachs and Souvaine~\cite{Sachs} and Hoffmann et al.~\cite{hoffmann1991art} showed that no art gallery problem with $n$ vertices and $h$ holes requires more than $\lfloor \frac{n+h}{3} \rfloor$ point guards and provided an O($n^2$) algorithm to find such placement, which is based on triangulation and 3-coloring. 

\subsection{Orthogonal Polygon with Holes Results}

For this version, all polygons and holes are orthogonal and axis-aligned. Let $orth(n, h)$ be the minimum number of point guards and $orth^v(n, h)$ be the minimum number of vertex guards necessary to guard any orthogonal polygon with \textit{n} vertices and \textit{h} holes. Note that $orth(n, h) \leq orth^v(n, h)$.

O'Rourke's method extends to show that: $orth^v(n, h) \leq \lfloor \frac{n + 2h}{4}  \rfloor$.
Shermer~\cite{artgallery} conjectured that $orth^v(n, h) \leq \lfloor \frac{n + h}{4}  \rfloor$ which Aggarwal~\cite{aggarwal1984art} established for $h = 1$ and $h = 2$.  Zylinski~\cite{zylinski2006orthogonal} showed that $\lfloor \frac{n + h}{4}  \rfloor$ vertex guards are always sufficient to guard any orthogonal polygon with \textit{n} vertices and \textit{h} holes, provided that there exists a quadrilateralization whose dual graph is a cactus.

O'Rourke also conjectured that $orth(n, h)$ is independent of $h$: $orth(n, h) = \lfloor \frac{n}{4} \rfloor$, which was verified by Hoffmann~\cite{hoffmann}. In 1990, Hoffmann~\cite{hoffmann} showed that $\lfloor \frac{n}{4} \rfloor$ point guards are always sufficient and sometimes necessary to guard an orthogonal polygon with $n$ vertices and an arbitrary number of holes. 
In 1996, Hoffmann and Kriegel~\cite{hoffmann1996graph} showed that $\leq \lfloor \frac{n}{3} \rfloor$ vertex guards are sufficient to watch the interior of an orthogonal polygon with holes. 

Consider  $orth^v(n, .)$ as the maximum of $orth^v(n, h)$ over all \textit{h}. Hoffmann conjectured that $orth^v(n, .) \leq \lfloor \frac{2n}{7} \rfloor$, disproving the earlier conjecture of Aggarwal~\cite{aggarwal1984art} that $orth^v(n, .) \leq \lfloor \frac{3n}{11} \rfloor$. In 2013, Michael and Pinciu~\cite{michaelguarding} improved this bound and showed that an orthogonal gallery with $n$ vertices and an unspecified number of holes can be guarded by at most $\frac{17n - 8}{52}$ vertex guards (17 / 52 = 0.3269).

In 1998, Abello et al.~\cite{abello1998illumination} provided a first tight bound of $\lfloor \frac{3n + 4(h - 1))}{8} \rfloor$ for the number of orthogonal guards placed at the vertices of an orthogonal polygon with $n$ vertices and $h$ holes which are sufficient for the cover the polygon and described a simple linear-time algorithm to find the guard placement for an orthogonal polygon (with or without holes). 

In 2016, Rezende et al.~\cite{de2016engineering}  showed the chronology of developments, and compared various current algorithms aiming at providing
efficient implementations to obtain optimal, or near-optimal, solutions.

\subsection{Families of Convex Sets (Triangles and Quadrilaterals) on the Plane Results}
\label{convexSets}

In 1977, Toth~\cite{toth1977illumination} considered the following problem: Given a set $F$ of \textit{n} disjoint compact convex sets in a plane, how many guards are sufficient to cover every point in the boundary of each set in $F$. Toth proved that $max\{2n, 4n-7\}$ point guards are always sufficient to cover $n$ disjoint compact convex sets in a plane. Everett and Toussaint~\cite{everett1990illuminating} proved that the families of \textit{n} disjoint squares $n > 4$,  can always be guarded with \textit{n} point guards. For families of disjoint isothetic rectangles (rectangles are \textit{isothetic} if all their sides are parallel
to the coordinate axes.), Czyzowicz et al.~\cite{czyzowicz1993illuminating} proved that $\lfloor \frac{4n+4}{3} \rfloor$ point guards suffice and conjectured that $n+c$ point guards would suffice, $c$ is a constant. If the rectangles have equal width, then $n+1$ point guards suffice, and $n-1$ point guards are occasionally necessary. Refer~\cite{martini1999survey} for more details.

In 1994, Blanco et al.~\cite{blanco1994illuminating} considered the problem of guarding the region of the plane, excluding the interior of the quadrilaterals (free space). Given $n$ pairwise disjoint quadrilaterals in the plane whose convex hull has no cut-off quadrilaterals, they showed that $2n$ vertex guards are always sufficient to cover the free space and all locations could be found on $O(n^2)$ time. If the quadrilaterals are isothetic rectangles, all locations can be placed in $O(n)$ time.

Urrutia~\cite{variation} showed that any family of \textit{n} disjoint rectangles can be guarded with at most $n+1$ point guards. A big rectangle encloses the elements of $F$, and consider this as an orthogonal polygon with holes. The total number of vertices is now $4n + 4$. Using the results from guarding orthogonal polygon with holes, this can be guarded with $n + 1$ guards.

Garcia-Lopez~\cite{de1995problemas} proved that $\lfloor \frac{5m}{9} \rfloor$ vertex lights are always sufficient and $\lfloor \frac{m}{2} \rfloor$ vertex guards are occasionally necessary to guard the free space generated by a family of disjoint polygons with $m$ vertices. To cover the free space generated by any family of $n$ disjoint quadrilaterals, he proved that $2n$ vertex lights are always sufficient and occasionally necessary
and that $\lfloor \frac{5n+3}{3} \rfloor$ point guards are always sufficient. He conjectured that $n + c$ point lights can always cover the free space generated by $m$ disjoint quadrilaterals, $c$ is a constant which was proved false by Czyzowicz and Urrutia~\cite{Czyzowicz1996}.

Czyzowicz et al.~\cite{czyzowicz1994protecting} proposed the following problem: Given a set $F$ of \textit{n} disjoint compact convex sets in a plane, how many guards are sufficient to protect each set in $F$. A set $F$ is protected by a guard $g$ if at least one point in the boundary of $F$ is visible from $g$. They prove that $\lfloor \frac{2(n-2)}{3} \rfloor$ point guards are always sufficient and occasionally necessary to protect any family of $n$ disjoint convex sets, $n> 2$. To protect any family of $n$ isothetic rectangles, $\lceil \frac{n}{2} \rceil$ point guards are always sufficient, and $\lfloor \frac{n}{2} \rfloor$ point guards are sometimes necessary.

Czyzowicz et al.~\cite{czyzowicz1993illuminating} showed that any family of \textit{n} disjoint triangles can be guarded with at most $\lfloor \frac{4n+4}{3} \rfloor$ point guards are sufficient and $n-1$ are occasionally necessary to guards. They also showed that  $n + 1$ guards are always sufficient and $n-1$ guards are occasionally necessary to illuminate any family of $n$ homothetic triangles and conjectured that there is a constant \textit{c} such that $n+c$ point guards sufficient to guard any collection of \textit{n} triangles. Later, Toth~\cite{toth2003guarding} showed that $\lfloor \frac{5n+2}{4} \rfloor$ guards can monitor the boundaries and the free space of \textit{n} disjoint triangles.

\subsection{Polyhedral Terrain Results}

A polyhedral terrain is a polyhedral surface in three dimensions such that its intersection with any vertical line is either empty or a point. A polyhedral terrain is triangulated if each of its faces is a triangle. Notice that a polyhedral terrain has a different structure than a city with vertical buildings. The results related to guarding polyhedral terrain focus on edge and face guards~\cite{everett1994edge,cole1,batista2010complexity,bose1,iwamoto2012lower,iwamoto2012improved}.

\subsection{Polyhedron Results}

A polyhedron in $R^3$ is a compact set bounded by a piece-wise linear manifold.
For guarding an \textit{orthogonal polyhedral}, points guards are less effective. There exist polyhedra with \textit{n} vertices where guards placed at every vertex do not cover the whole interior and $O(n^{3/2})$ non-vertex guards are required for interior coverage~\cite{artgallery}.  Thus, results related to guarding $R^3$ polyhedrons focus on edge and face guards~\cite{variation,cano2012edge,souvaine2011face,viglietta2015reprint,benbernou2011edge,viglietta2017optimally}.

\subsection{City Guarding Results}

In 2008, Bao et al.~\cite{cityguarding} proposed the city guarding problem where one is given a city with $k$ vertical buildings, each having an axis-aligned rectangular base. The guards are to be placed only at the top vertices of the buildings. They showed that $\lfloor \frac{2(k-1)}{3} \rfloor + 1$ vertex guards are sometimes necessary and always sufficient to guard the roofs (Roof Guarding Problem). 
They further proved that $k + \lfloor \frac{k}{4} \rfloor + 1$ vertex guards are always sufficient to guard the ground and the walls, and $k + \lfloor \frac{k}{2} \rfloor + 1$ vertex guards are always sufficient to guard the aerial space, which includes all roofs and walls of the buildings, and the ground. Their results directly apply to problem (V1) and imply that $2k + 2\lfloor \frac{k}{4} \rfloor + 2$ vertex guards with $180^{\circ}$ vision are always sufficient to guard the walls and ground, and $2k + 2\lfloor \frac{k}{2} \rfloor + 2$ are needed to guard the city.

It follows from Table~\ref{t:results} that our results are a significant improvement over those that can be inferred   from~\cite{blanco1994illuminating,cityguarding}.
Due to space constraints, we refer the reader to Appendix~\ref{RW} for detailed related work.

We start with the following theorem, that allows us to limit our attention to guarding the roofs and walls of the buildings, and the ground.

\begin{theorem}
If guards are placed so that the roofs, walls, and the ground of the city are guarded, then every point in the aerial space of the city is guarded.
\end{theorem}

\begin{proof}
Let $p$ be a point in the aerial space of the city and assume $p$ is not guarded. Let $p’ $ be the vertical projection of $p$ onto the ground (or a building roof) and let $g$ be a guard that sees $ p’ $ (such $g$ exists since the ground of the city is guarded). Then $g$, $p$ and $p'$ define a vertical plane $\pi$. Consider the vertical triangle defined by $g$, $p$, and $p'$. If any portion of a building intersects the triangle side $ gp $ at some point $q$ then the line segment $qq’ $ is part of that building, where $q’ $ is the vertical projection of $q$ onto the ground. Since the line segment $qq'$ intersects the triangle side $gp'$ it then follows  that $p'$ is not visible from $g$, a contradiction.       
\end{proof}

A similar result holds if we aim for walls and ground guarding only (no roof guarding requirement), including the space between the buildings.

\section{Axis-Aligned Rectangle-Base Buildings}
\label{s:AARB}

Given a rectangular city with \textit{k} disjoint vertical buildings, each having an axis-aligned rectangular base, the goal is to place the minimum number of cameras that can see only the half-space in front of them  (denoted as 180$^{\circ}$ range of vision), at the top corners (vertices) of the buildings to guard the city (roofs, walls, ground, and aerial space). 
Thus, when a guard (camera) is aligned with a wall of the building it is placed on the half-space seen by the guard is bounded by a vertical plane containing that wall.

\subsection{Roof Guarding}
\label{ss:RGAARB}

\begin{theorem}
\label{t:RGAARB}
Given a city with $k$ disjoint axis-aligned rectangular buildings, $k$ vertex guards are always sufficient and sometimes necessary to guard the roofs.

\end{theorem}

\begin{proof}
The sufficiency bound is trivial. For the necessary part, consider a set $S = \{B_1, B_2, B_3, \dots, B_k \}$ of $k$ buildings, shown in Figure~\ref{SetUp}, with the following setup: 
\begin{enumerate}
\item the height $h_{B_i}$ of building $B_i$ is greater than the height $h_{B_{j}}$ of building $B_{j}$, $\forall \hspace{3pt}i , j $ such that $ 1 \leq i < j \leq k$, and 
\item $\forall$ $i < j-1$, building $B_{j-1}$ totally blocks the visibility between $B_{i}$ and $B_{j}$.
\end{enumerate}

\begin{figure}[h]
    \centering
    \begin{subfigure}[t]{0.5\textwidth}
        \centering
        \includegraphics[width = \linewidth]{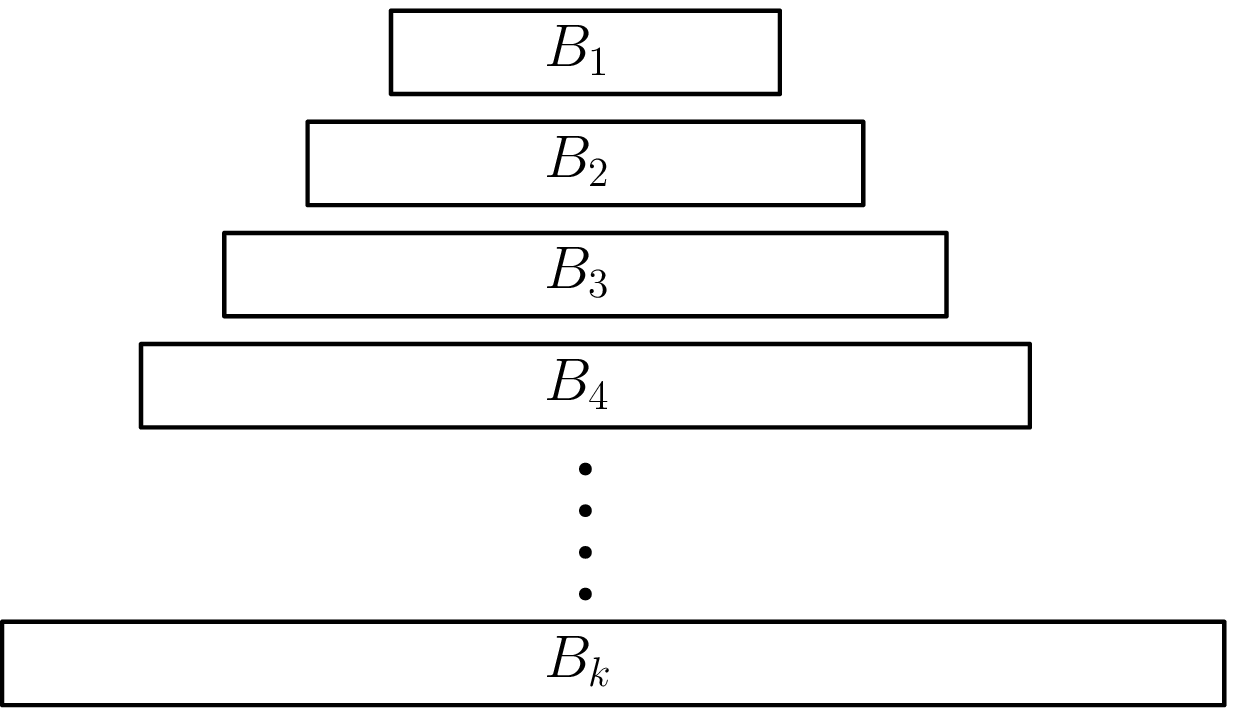}
        \caption{City Setup}
        \label{SetUp}
    \end{subfigure}
    \begin{subfigure}[t]{0.5\textwidth}
        \centering
        \includegraphics[width = 0.5\linewidth]{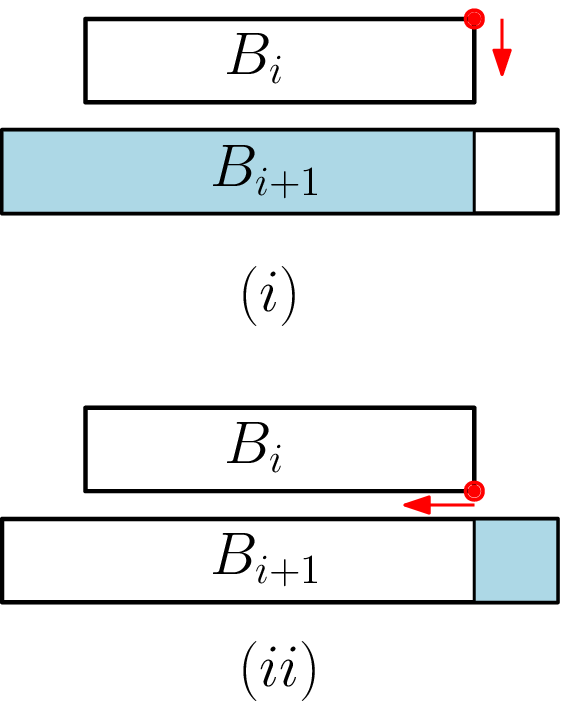}
        \caption{Possible guards position}
        \label{SetUp2}
    \end{subfigure}
    \caption{$k$ guards are needed to guard the roofs.}
    \label{f:roof1}
\end{figure}

We need to place the first guard on building $B_1$ to guard its roof, because $h_{B_1} > h_{B_i},  \forall i > 1$. There are four possible positions to place a vertex guard on building $B_1$. Let the guard be placed on one of the right vertices (placing the guard on one of the left vertices results in symmetric cases). Consider the two relevant orientations of the guard, shown in Figure~\ref{SetUp2}, out of three possible positions that can see the roof of $B_1$, where the arrow corresponds to the direction in which the guard is guarding the roof. Notice that a guard facing West is placed at the lower vertex of $B_1$ rather than at the top vertex. Due to the limited visibility of the guard, there is no vertex on building $B_1$ from where roofs of both building $B_1$ and $B_2$ are completely visible. Therefore, the next guard should either be placed on building $B_2$ or on $B_1$, such that the roof of building $B_2$ is completely visible after placing this guard. If we place the second guard on $B_1$, then the next guard must be placed on building $B_3$ because no point on the roof of building $B_3$ is visible by the previously placed guards. Thus, we can place the next guard on building $B_2$.  
The rest follows by induction on the number of buildings, as we are left with a similar problem on $k-1$ buildings.
\end{proof}


\subsection{Ground and Wall Guarding}
\label{ss:GWAA}


Vertically projecting the city on the ground results in a rectangle polygon with $k$ rectangular holes. As mentioned earlier, the number of guards required to guard the walls and ground is no larger than the number of guards needed for the following problem: 
\begin{subproblem}
\label{axisAlignedRectangle}
(V1) {\it Given an axis-aligned rectangle $P$ with $k$ disjoint axis-aligned rectangular holes, place vertex guards on hole boundaries such that every point inside $P$ is visible to at least one guard, where the range of vision of guards is $180^{\circ}$}. 
\end{subproblem}

On the other hand, it is easy to see that a lower bound on the number of guards for Subproblem~\ref{axisAlignedRectangle} can be used to obtain a lower bound for guarding the walls and the ground of a city with $k$ rectangular buildings: map the holes to buildings of the same height.

It is worth noticing though that the two problems are not equivalent, that is, for a given input, fewer guards might be needed to guard the walls and ground than the number needed to guard the holes defined by projecting the buildings to the ground.

Observe that $2k + \lfloor \frac{k}{2} \rfloor + 2$ vertex guards, placed on holes, can be obtained from~\cite{cityguarding} by replacing a $360^{\circ}$ guard with two $180^{\circ}$ guards. In what follows, we show how to improve this bound. For each hole, extend the right vertical edge in the upward (North) direction through the interior of the polygon until it encounters some horizontal edge of a hole or the outer rectangle. After extending the vertical edges,  extend both horizontal edges of each hole in the left (West) direction through the interior of the polygon until it encounters some vertical edge or extended vertical edge of a hole or the outer rectangle.
The steps above divide the polygon into 2$k$+1 shapes. Each shape corresponds to a monotone staircase (both in $x$ and $y$-direction). 
Only one guard is required to guard each staircase, placed at the South-East corner of the staircase, facing West. Out of $2k + 1$ guards, $2k$ guards are placed on the vertices of the holes while one guard is placed on a vertex of the rectangle $P$. Refer to Figure~\ref{f:extension} for visual details. 

\begin{figure}[b]
    \centering
        \includegraphics[width = \linewidth]{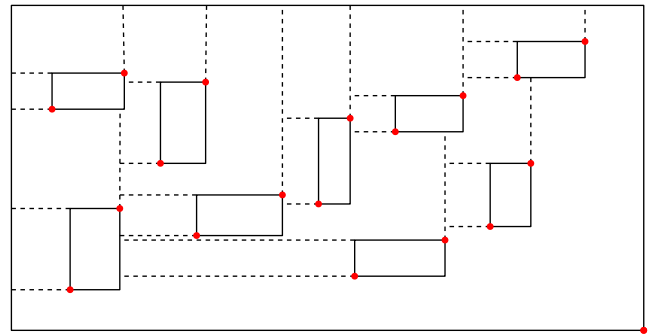}
    \caption{Extension of the right vertical edge and horizontal edges of each hole. $2k + 1$ guards are always sufficient to guard the walls and ground.}
    \label{f:extension}
\end{figure}

\begin{theorem}
\label{1+2k}
$2k+1$ guards are always sufficient to guard the walls and ground of a rectangular city with $k$ disjoint axis-aligned rectangular buildings, with at most one guard placed at a corner of the bounding rectangle. 
\end{theorem}
If however, we do not allow a guard to be placed at a corner of the enclosing rectangle $P$, then the number of guards needed could increase significantly. 
In the rest of this section, we prove an upper bound on the number of vertex guards, placed only on vertices of the holes (the setup in~\cite{blanco1994illuminating,cityguarding}). 

Let $S$ be the set of $k$ buildings in the city, contained in the axis-aligned rectangle $P$
defined by the points $[0, 0; x, y]$. Let $x_s^M, x_f^M$ be the starting and finishing boundry sequence along $x$-axis and $y_s^M, y_f^M$ be the starting and finishing boundry sequence along $y$-axis. We define four types of staircases (see Figure~\ref{types}):
(i) Rising staircase (RS): $x_s^M = 0$, $y_f^M = y$, $x_f^M$ is non decreasing along the positive $y$-axis, and $y_s^M$ is non decreasing along the positive $x$-axis
(ii) Falling staircase (FS):    $x_f^M = x$,  $y_f^M = y$, $x_s^M$ is non increasing along the positive $y$-axis, and $y_s^M$ is non increasing along the positive $x$-axis
(iii) Reverse rising staircase (RRS): $x_f^M = x$, $y_s^M = 0$, $x_s^M$ is non decreasing along the positive $y$-axis, and $y_f^M $ is non decreasing along the positive $x$-axis
(iv) Reverse falling staircase (RFS): $x_s^M = 0$,  $y_s^M = 0$, $x_f^M$ is non increasing along the positive $y$-axis, and $y_f^M$ is non increasing along the positive $x$-axis

\begin{figure}[t]
    \centering
    \begin{subfigure}[t]{0.5\linewidth}
        \centering
        \includegraphics[width = 0.8\linewidth]{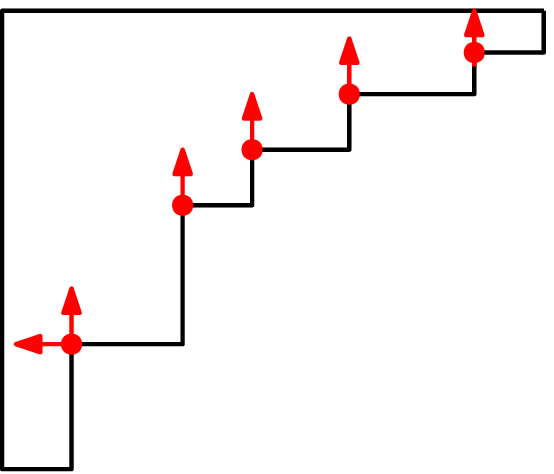}
        \caption{Rising Staircase.}
        \label{RS}
    \end{subfigure}%
    \begin{subfigure}[t]{0.5\linewidth}
        \centering
        \includegraphics[width = 0.8\linewidth]{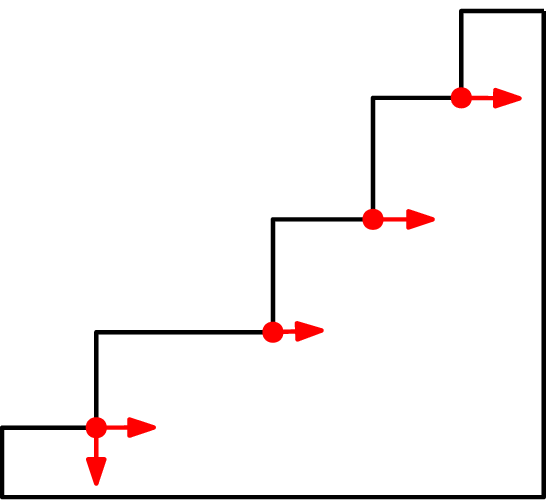}
        \caption{Reverse Rising Staircase.}
        \label{RRS}
    \end{subfigure}%
    
    \begin{subfigure}[t]{0.5\linewidth}
        \centering
        \includegraphics[width = 0.8\linewidth]{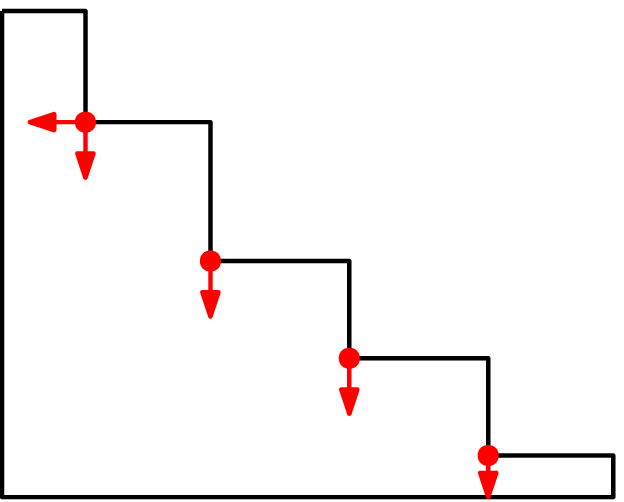}
        \caption{Reverse Falling Staircase}
        \label{RFS}
    \end{subfigure}%
     \begin{subfigure}[t]{0.5\linewidth}
        \centering
        \includegraphics[width = 0.8\linewidth]{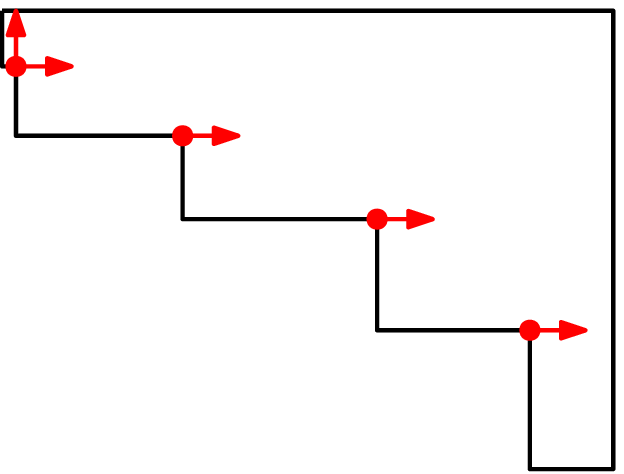}
        \caption{Falling Staircase}
        \label{FS}
    \end{subfigure}%
    \caption{Type of staircases and placement of guards.}
    \label{types}
\end{figure}
 
A rising staircase is constructed as follows: extend the horizontal edges of each hole towards the right (East) direction, then extend the vertical edges towards the South direction. The closed orthogonal polygon formed by the top edge and the left edge of $P$, and the extended edges of the holes, corresponds to a rising staircase.  
Falling, reverse rising, and reverse falling staircases are constructed similarly. 
Note that for each staircase a reflex vertex corresponds to a vertex of a hole. Thus, the number of buildings involved in a staircase is equal to the number of reflex vertices on the staircase. 

We find the staircase comprising the minimum number of buildings. WLOG assume the staircase involving the minimum number of buildings is $RRS$ (otherwise, we can rotate the input so that the staircase corresponds to RRS). For this staircase, place a guard on each reflex vertex, facing right (East), and an additional guard on the first (bottom) stair, with the guard facing down (South), as shown in Figure~\ref{types}. 
The number of guards required to cover the staircase is one more than the number of steps (stairs) in it.
In the worst case, each of the four staircases must have the same number of stairs, otherwise we can use one with the smallest number as RRS. 
  
\begin{figure}[h]
\centering
\includegraphics[width= \linewidth]{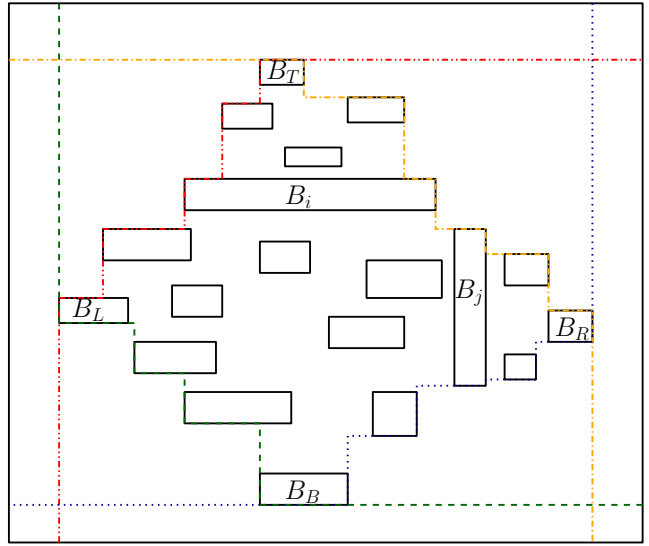}
  \captionof{figure}{Staircases $RS$ (dash dot / red), $FS$ (dash dot / orange), $RRS$ (dotted / blue) and $RFS$ (dashed / green). Buildings above $B_i$ lie in its vertical span and buildings right of $B_j$ lie in its horizontal span.}
  \label{span}
\end{figure}
In what follows, we provide a divide and conquer approach to find an upper bound on the number of guards.
For some building $B$, the vertical span of $B$ is the parallel strip defined by the vertical sides of $B$ and containing $B$. The horizontal span is defined accordingly. 
Let $B_L$ be the leftmost building, $B_T$ be the topmost building, $B_R$ be the rightmost building, and $B_B$ be the bottom-most building of the city. A building $B_i$ is called an \textit{internal building} if $B_i \not \in \{B_L, B_T, B_R, B_B\}$. The pair of staircases (i) $RS, FS$ (ii) $FS, RRS$ (iii) $RRS, RFS$ and (iv) $RFS, RS$ are called \textit{adjacent staircases} while the pairs (v) $RS, RRS$, and (vi) $FS, RFS$ are called \textit{opposite staircases}.

Assume two adjacent staircases, say $RS$ and $FS$, share the same building $B_i \not \in \{B_L, B_T, B_R, B_B\}$. Then, all buildings that lie in the upper half-plane defined by the line supporting the upper horizontal edge of $B_i$ (buildings above $B_i$) are in the vertical span of $B_i$ (see Figure~\ref{span}). Similarly, if staircases $RRS$ and $FS$ include the walls of the same building $B_j \not \in \{B_L, B_T, B_R, B_B\}$, then all buildings that lie on the right of $B_j$ are in the horizontal span of $B_j$.

Notice it is possible that, from the set of building pairs $(B_L, B_T), (B_T, B_R), (B_R, B_B)$, and $(B_B, B_L)$, the pair in one of the sets corresponds to the same building.
In this situation (call it Case 0), one of the staircases consists of only one stair, and two guards are required to guard the staircase. Using a placement of guards like in Theorem~\ref{1+2k}, $2k + 2$ guards are required to cover the walls and ground of such a rectangular city.
Thus, from now on, we assume this is not the case.

Consider the four staircases $RS, FS, RFS$, and $RRS$. We can have four cases:

\noindent \textbf{Case 1}: No adjacent pair of staircases share the same internal building, and no opposite pair of staircases share the same building.

WLOG assume that $RRS$ contains the minimum number of stairs; place the guards in a similar fashion as in Theorem~\ref{1+2k}. Overall, we place two guards on each building and the rest on the reflex vertices of the staircase $RRS$. 
 
The upper bound on the number of vertex guards required to cover the staircase of $P$ is achieved when the number of buildings involved in the construction of each staircase is the same. Let the staircases $RS$ and $RRS$ contain $\delta$ distinct buildings. Staircase $FS$ contains $\delta - 2$ distinct buildings because building $B_T$ is already counted in staircase $RS$ and building $B_R$ is counted in staircase  $RRS$. Similarly, $RFS$ contains $\delta - 2$ distinct buildings.  Note that $\delta + \delta + \delta - 2 + \delta - 2 = k$ and thus $\delta = \lfloor \frac{k}{4} \rfloor  + 1$. Therefore, to guard each staircase, we require $\delta + 1 = \lfloor \frac{k}{4} \rfloor + 1 + 1 = \lfloor \frac{k}{4} \rfloor + 2$ guards. 

We place $2$ guards on each building and $\lfloor \frac{k}{4} \rfloor + 2$ guards to cover the staircase. Therefore, $2k +\lfloor \frac{k}{4} \rfloor + 2$ guards are required to cover the walls and ground. 

\noindent \textbf{Case 2}: At least one pair of opposite staircases shares the same building, and no adjacent staircases share the same interior building.

Let the staircases $RS$ and $RRS$ share a building $B_i$. Refer to Figure~\ref{situation3} and note that extending the top edge of $B_i$ to the left until it hits $P$ will not result in an intersection with the other buildings. Similarly, extending the bottom edge of $B_i$ to the right until it hits $P$ will not result in an intersection with the other buildings. We divide the city into two sub-cities, $city_1$ and $city_2$ (green and orange boundaries in Figure~\ref{situation3}), by extending the top edge of $B_i$ towards left and the bottom edge towards the right. Let $B_i$ be included in both sub-cities.

All buildings in $city_1$ lie either above or towards the right of $B_i$. We place two guards on each building, one at the North-West corner and one at the South-East corner, both facing East. We further place a third guard on the North-West corner of $B_i$, facing West. 
All buildings in $city_2$ lie either below or towards the left of $B_i$. We place two guards on each building, one at the South-East corner and one at the North-West corner, both facing West. We further place an additional guard on the South-East corner of $B_i$ facing East. In total, we have placed six guards on $B_i$. However, two guards are duplicates, so we only have four guards on $B_i$.  
Thus, $2k + 2$ guards are required to cover the walls and ground of the city. 
\begin{figure}[h]
\centering
\includegraphics[width=\linewidth]{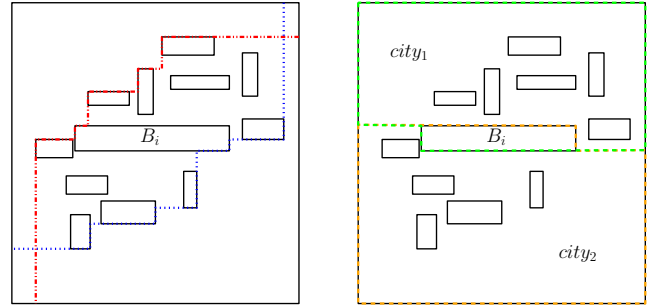}
\caption{Building $B_i$ is shared by staircases $RS$ and $RRS$. Staircase $RS$ is shown in dash dot / red color and staircase $RRS$ is shown in dotted / blue color.}
\label{situation3}
\end{figure}

\noindent \textbf{Case 3}: At least one pair of adjacent staircases shares the same interior building, and no pair of opposite staircases  shares the same building.

We use the following recursive approach to compute an upper bound on the number of guards. Let $B_i$ be a building shared by two adjacent staircases, say $RS$ and $FS$, as shown in Figure~\ref{DivisionOfCity}. Let $\alpha_i$ be the number of buildings that lie above $B_i$ and $\beta_i$ be the number of buildings (excluding $B_i$), that lie in the lower half-plane defined by the line supporting the upper horizontal edge of $B_i$. Note that $\alpha_i + \beta_i + 1 = k$. Let $C$ be the set containing all such buildings $B_i$ (walls included in more than one staircase). Let $B_j \in C$ be the building that minimizes the value $|\alpha_j - \beta_j|$, such that $\alpha_j, \beta_j \geq 3$.

If such building does not exist then each building $B_j$ in set $C$ has $\alpha_j \leq 3$ or $ \beta_j \leq 3$. We can use a similar argument as the one discussed in Case 2. Recall that in the worst case all staircases should have an equal number of buildings/stairs. It is easy to notice that there exist at most three buildings shared by a staircase  pair (i) $RS, FS$ (ii) $FS, RRS$ (iii) $RRS, RFS$, or (iv) $RFS, RS$, as $\alpha_i < 3$ or $\beta_i < 3$. Let each of $RS, RRS$ contain $\delta$ distinct buildings. Staircases $FS, RFS$ contain $\delta - 6$ distinct buildings, as three buildings are included in each of $RS$ and $RRS$. There are $k$ buildings, $2 \times \delta + 2 \times (\delta - 6) = k$, thus $4 \times \delta - 12 = k$, and   $\delta = \lfloor \frac{k}{4} \rfloor + 3$. To guard the staircase we need at most $\lfloor \frac{k}{4} \rfloor + 4$ guards, resulting in $2k + \lfloor \frac{k}{4} \rfloor + 4$ guards overall.

If there exists a building $B_j$ such that $\alpha_j, \beta_j \geq 3$, we proceed as follows.
Let the staircases $RS$ and $FS$ share building $B_j$.  There can be two cases: (i) there are no buildings within the horizontal span of $B_j$ (refer to Figure~\ref{cityDivide}) and (ii) there exist buildings within the horizontal span of $B_j$ (refer to Figure~\ref{cityDivide2}).  

\begin{figure}[t]
\centering
  \subcaptionbox{No buildings within the horizontal span of $B_j$\label{cityDivide}}{\includegraphics[width=\linewidth]{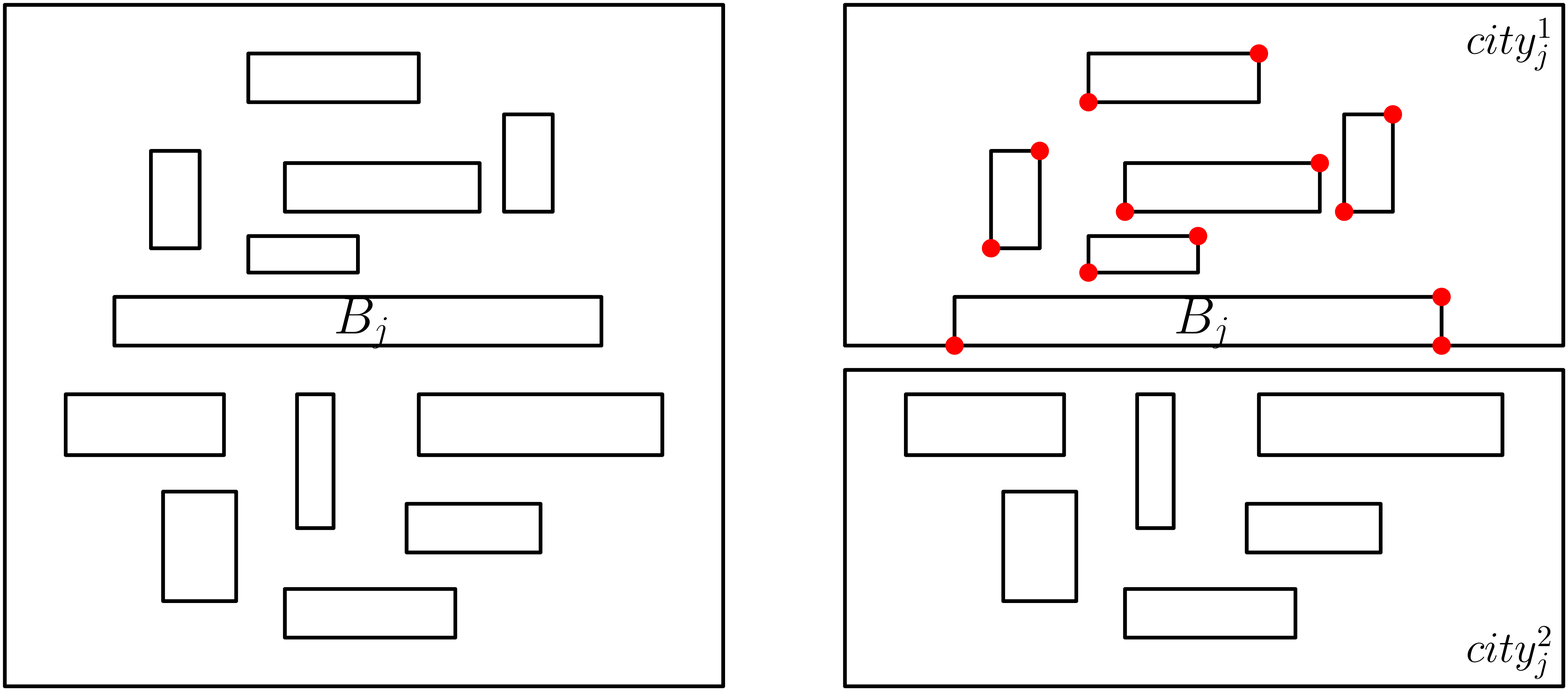}}\hspace{1em}
  
  \subcaptionbox{There exist buildings within the horizontal span of $B_j$\label{cityDivide2}}{\includegraphics[width=\linewidth]{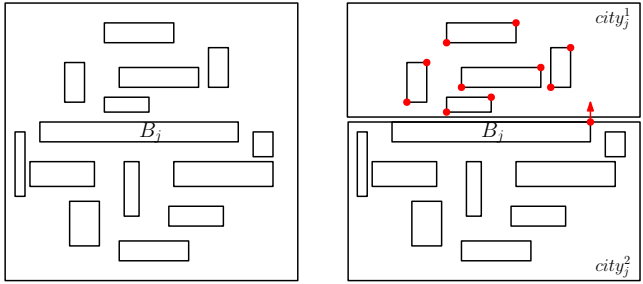}}
    \caption{Building $B_j$ shares its walls with $RS$ and $FS$. Divide city into two sub-cities, $city_j^1$ and $city_j^2$.}
    \label{DivisionOfCity}
\end{figure}

Consider the first case where none of the buildings lie within the horizontal span of $B_j$. We divide the city into two sub-cities, $city_j^1$ with $\alpha_j$ buildings and $city_j^2$ with $\beta_j$ buildings. All the buildings in one of the sub-cities lie inside the vertical span of $B_j$. Let this sub-city be $city_j^1$. We add $B_j$ to $city_j^1$, which results in a total of $(\alpha_j + 1)$ buildings in  $city_j^1$, and follow Case 0, which results in a total of $2 (\alpha_i + 1) + 1$ guards to guard the walls and ground of this sub-city, as shown in Figure~\ref{cityDivide}. For $city_j^2$, we consider the Case it falls in and place the guards accordingly. For the placement of guards, we treat the two sub-cities as independent cities. It is important to notice that only one of $city_j^1$ and $city_j^2$ above can be in Case 3, while the other one is in Case 0. Thus, only one of the two cities could need further divisions. 

Assume Case 3 keeps occurring, and we need to divide the city \textit{m} times.
During each division, the sub-city with corresponding building $B_j$ contains greater than or equal to four buildings and one additional guard is required to guard such sub-city. Let the first division divide the city into two sub-cities with $k_1, k - k_1$ buildings each. The second division splits the sub-city with $k - k_1$ buildings into $k_2, k - k_1 - k_2$ buildings and so on, down to sub-cities with $k_m, k - \sum_{i = 1}^{m} k_i$ buildings. Each resulting sub-city does not need to be divided further. Let $k' = \sum_{i = 1}^{m} k_i$. For each sub-city with $k_1, k_2, \dots k_m$ buildings, we require $2k_i + 1$ guards, where $k_i \geq 4$, and for the last sub city we need at most $4 + \lfloor (k - k')/4 \rfloor$ additional guards. Thus, the total number of guards required to guard the city is:

\noindent $2k_1 + 1 + 2k_2 + 1 + \dots + 2k_m + 1 + 2(k - k') + 4 + \lfloor \frac{ k - k'}{4} \rfloor = 2k + m + 4 + \lfloor \frac{ k - k'}{4} \rfloor \leq 2k + 4 + k'/4 +  \lfloor \frac{ k - k'}{4} \rfloor \leq  2k + \lfloor \frac{k}{4} \rfloor + 4$

Consider the second case where buildings lie within the horizontal span of $B_j$. We divide the city into two sub-cities, $city_j^1$ with $\alpha_j$ buildings and $city_j^2$ with $\beta_j$ buildings, such that all buildings in one of the sub-cities lie inside the vertical span of $B_j$. Let this sub-city be $city_j^1$. We add $B_j$ to $city_j^2$, which results in a total of $(\beta_j + 1)$ buildings in $city_j^2$. For the placement of guards, $city_j^1$ is in Case 0, and we place two guards on each building of $city_j^1$, and one guard on building $B_j$ facing towards $city_j^1$, which results in a total of $2 (\alpha_i) + 1$ guards to cover the walls and ground of $city_j^1$, as shown in Figure~\ref{cityDivide2}. For $city_j^2$, we consider the case it falls in and places the guards accordingly. For the placement of guards, we treat the two sub-cities as independent cities. Using a similar explanation as in the first case, we obtain the upper bound of $2k + \lfloor \frac{k}{4} \rfloor + 4$.

\noindent \textbf{Case 4}: At least one pair of opposite staircases and one pair of adjacent staircases share the same building.

We place guards according to Case 2 and conclude that  $2k + 2$ guards are sufficient to cover the walls and ground of the city.

The derived upper bound obviously holds when  
the guards can have {\it arbitrary} orientation and we have:  
\begin{theorem}
\label{th-aa}
$2k+\lfloor \frac{k}{4} \rfloor +4$ guards are always sufficient to guard the walls and ground of a rectangular city with $k$ disjoint axis-aligned rectangular buildings, with all guards placed on vertices of the buildings. 
\end{theorem}

We also obtain the following Theorem, restricting the visibility of guards in~\cite{blanco1994illuminating} to $180^{\circ}$:
\begin{theorem}
Given $k$ pairwise disjoint isothetic rectangles in the plane, $2k+\lfloor \frac{k}{4} \rfloor +4$ vertex guards are always sufficient to guard the free space. 
\end{theorem}

\subsection{City Guarding}
\label{ss:CGAA}

\begin{theorem}
\label{CGFT}
Given a city with $k$ disjoint axis-aligned buildings within an axis-aligned rectangle $P$, the number of axis-aligned cameras with 180$^{\circ}$ range of vision needed to guard the city (roofs, walls, and ground) is upper bounded by (i) $2k+\lfloor \frac{k}{4} \rfloor +4$ when cameras are placed on buildings only, and (ii) $2k+1$ when at most one camera can be placed at a corner of $P$.
\end{theorem}
\begin{proof}
The solution described earlier in Section~\ref{ss:GWAA} established either $2k+\lfloor \frac{k}{4} \rfloor +4$ guards in case (i) and or $2k+1$ guards in case (ii) to cover the ground and the walls of the city. In both (i) and (ii), on each building, we place at least two guards, on diagonal corners, facing in the same direction. Thus, one of these two guards also guards the roof of the building. Therefore, $2k+\lfloor \frac{k}{4} \rfloor +4$ guards are always sufficient to guard the city in case (i) and $2k+1$ in case (ii).  
\end{proof}

\section{Arbitrary Oriented Rectangle Buildings}
\label{s:AORB}

Given a rectangular city with \textit{k} vertical buildings, each having a rectangular base, the goal is to place cameras with 180$^{\circ}$ range of vision, at the top corners (vertices) of the buildings, to guard the city. In all our proofs, each guard is aligned with a wall of the building it is placed on, similar to the axis-aligned version. 
The same city structure in Theorem~\ref{t:RGAARB} leads to:

\begin{theorem}
Given a city with \textit{k} disjoint rectangular buildings, \textit{k} vertex guards are always sufficient and sometimes necessary to guard the roofs.
\end{theorem}

As before, the problem of guarding the ground and the walls reduces to: 
\begin{subproblem}
(V2) {\it Given an axis-aligned rectangle $P$ with $k$ disjoint rectangular holes, place vertex guards on hole boundaries such that every point in $P$ is visible to at least one guard, where the range of vision of guards is $180^{\circ}$}.
\end{subproblem}


\begin{theorem}
\label{AARO}
$3k + 1$ vertex guards are sometimes necessary to guard a rectangular polygon \textit{P} with \textit{k} disjoint rectangular holes, where guards are placed only at vertices of the holes. We conjecture the bounds are tight. 
\end{theorem}

\begin{proof}
For the necessity part, consider the input in Figure~\ref{ncity1}, with the following properties:
\begin{enumerate}
\item $B_i$ lies within the span of $B_j \hspace{2mm}, \forall j < i$.
\item  None of the edges of $B_i$ is partially or completely visible from any vertex of $B_j, \hspace{2mm} \forall \hspace{2mm} j < i - 1$ and from any vertex of $B_m, \hspace{2mm} \forall \hspace{2mm} m > i+1$.
\item From each potential position of a vertex guard on $B_i$, the guard is able to see at most one edge of $B_{i+1}$.
\item There is no guard position on $B_i$ from where an edge of $B_{i-1}$ and an edge of $B_i$ are visible. 
\end{enumerate}

\begin{figure}[h]
\centering
\includegraphics[width = \linewidth]{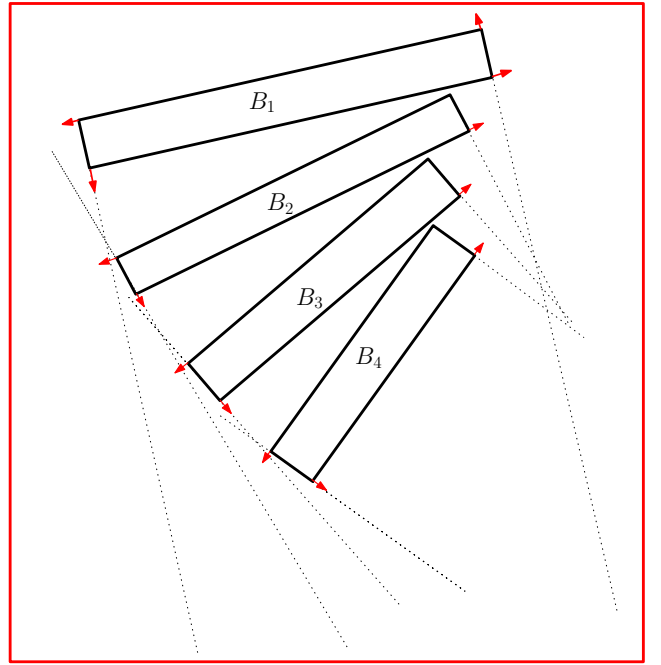}
\caption{City Structure where $3k+1$ guards are necessary to guard the polygon}
\label{ncity1}
\end{figure}

Consider the space $\wp_i$ between two consecutive holes $B_i$ and $B_{i+1}$, as shown in Figure~\ref{spaceBetween}. Because of property 2, $\wp_i$ is not visible to any guard placed on building $B_j$ where $j \in [1, i) \cup (i+1, k]$. Therefore, $\wp_i$ is only visible to the guards either placed on $B_i$ or $B_{i+1}$. It is easy to notice that there are twelve possible guard positions on $B_i$ and $B_{i+1}$ from where $\wp_i$ is visible (partially from each position, see Figure~\ref{spaceBetween}), and these guards cover three walls, one wall of $B_i$ and two walls of $B_{i+1}$. Note that these guards do not cover any other wall (partially or entirely), and the mentioned three walls are not visible (partially or entirely) to any other potential guard. Out of twelve possible guard positions, the minimum number of guards required to guard $\wp_i$ is two (either both placed on $B_i$ or one placed on $B_i$ and another on $B_{i+1}$. Therefore the space between any two consecutive holes is only guarded by the guards placed on these holes, and the minimum number of guards required to guard such space is two. 

Consider the left wall, $w_L^i$ of hole $B_i$. Because of the structure of the city, $w_L^i$ is not visible to any guard placed on $B_j$ for $j \not = i$. Hence, $w_L^i$ can only be guarded by a guard placed on $B_i$, and there are four possible guard positions from where $w_L^i$ is visible(see Figure~\ref{spaceBetween}). However, none of these guards positions cover any other wall in the city. Therefore, we need one guard to cover the left wall of each hole.

\begin{figure}[h]
\centering
\includegraphics[width = \linewidth]{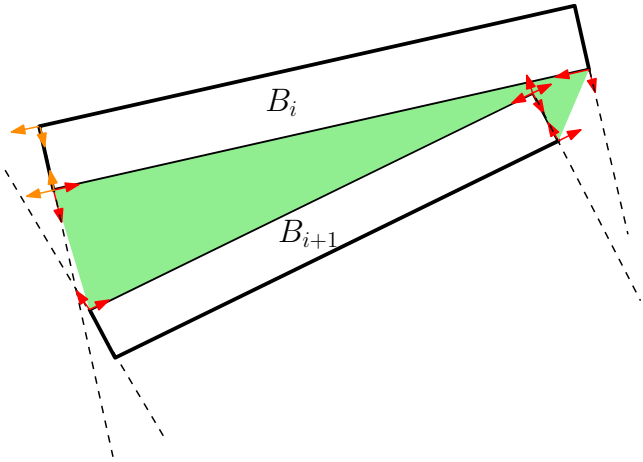}
\caption{$\wp$ is shaded in green. Potential guard positions to cover $\wp$ are shown in red and potential guard positions to cover the $w_L^i$ is shown in orange.}
\label{spaceBetween}
\end{figure}

There are $k-1$ spaces $\wp_i$ in total, between consecutive buildings ($i=1,2,\ldots,k-1$). Thus, $2(k-1)$ guards are needed to guard their union. As argued, each left edge of a hole needs an additional guard, resulting in $3k-2$ guards. The top and right edges of $B_1$ and the bottom edge of $B_k$ are not guarded, so in total, at least $3k+1$ guards are needed.

A $3k+1$ guard placement, for $k=4$, is shown in Figure~\ref{ncity1} and is obtained as follows.
We start placing guards on $B_1$. Three of its walls (left, top and right) are not visible by any potential guard placed on $B_i \hspace{2mm}, \forall i > 1$ (property 1). We place three guards to cover these walls. Consider the space $\wp_1$ between $B_1$ and $B_2$. We need two guards to cover $\wp_1$; let one of these guards be placed on $B_1$ and the other on $B_2$ (see Figure~\ref{ncity1}). After placing these two guards, all walls of $B_1$ are visible, and two walls of $B_2$ (top and right) are visible. We need an additional guard to cover the left wall of $B_2$, and this guard does not cover any other wall in the city. Consider now the space $\wp_2$ between the hole $B_2$ and $B_3$ and place two guards to cover $\wp_2$; let one of these guards be placed on $B_2$ and the other on $B_3$. After placing these two guards, all walls of $B_1$ and $B_2$ are guarded, and two walls of $B_3$ are guarded. We placed four guards on $B_1$ and three guards on $B_2$. We continue this process, and three guards are required to guard each building $B_i \hspace{2mm} \forall i > 2$. This results in $3k+1$ guards as we place three guards on each building $B_i \hspace{2mm}, \forall i \in (1, k]$, and four guards on $B_1$. Guard locations and directions are shown in Figure~\ref{ncity1}. 
\end{proof}

\noindent For guarding the city we have:

\begin{theorem}
\label{AARO-1}
$3k + 1$ vertex guards are sometimes necessary to guard a city with \textit{k} vertical buildings with rectangular base, where guards are placed only at the top vertices of the buildings. We conjecture the bounds are tight. 

\end{theorem}

\bibliography{references}

\appendix

\end{document}